\documentclass[aps,prd,nofootinbib,preprintnumbers,nofootinbibn,twocolumn]{revtex4}
\usepackage{graphicx}  
\usepackage{dcolumn}   
\usepackage{bm}        
\usepackage{amssymb}   
\usepackage[latin1]{inputenc}
\usepackage[english]{babel}
\usepackage{shapepar}

\usepackage{amsmath}
\usepackage{amsthm}
\usepackage{array}
\usepackage{fancybox}
\usepackage{multirow}
\usepackage{appendix}
\usepackage{epsfig}

\newtheorem{thm}{Theorem}[section]
\newtheorem{prop}[thm]{Proposition}
\theoremstyle{definition}
\newtheorem{defn}[thm]{Definition}
\begin{document}


\title{Singularities and n-dimensional black holes in torsion theories\\[0.4cm]}
\thispagestyle{empty}

\author{
J. A. R. Cembranos,\footnote{Electronic address: cembra@fis.ucm.es}
J. Gigante Valcarcel, \footnote{Electronic address: jorgegigante@ucm.es}
F. J. Maldonado Torralba, \footnote{Electronic address: fmaldo01@ucm.es}
\vspace*{0.3cm}}

\affiliation{Departamento de F\'{\i}sica Te\'orica I, Universidad Complutense de Madrid, E-28040 Madrid, Spain.}

\begin{abstract}
In this work we have studied the singular behaviour of gravitational theories with non symmetric connections.
For this purpose we introduce a new criteria for the appearance of singularities based on the existence of black/white hole regions of arbitrary codimension defined inside a spacetime of arbitrary dimension. 
We discuss this prescription by increasing the complexity of the particular torsion theory under study. In this sense, we start with Teleparallel Gravity,
then we analyse Einstein-Cartan theory, and finally dynamical torsion models. 
\end{abstract}
\maketitle

\section{Introduction}
\setcounter{page}{1}
In a physical theory, a singularity is commonly known as a ``place'' where some of the variables used in the description of the system diverge. For example, we find this  in the singularity in $r=0$ of the Coulombian potential $V=K\frac{q}{r}$. This kind of  behaviour appears mainly because the theory is not valid in the considered region or we have assumed a simplification. In the previous example the singularity arises due to the fact that we are considering the charged particle as a point and neglecting the quantum effects.

In General Relativity (GR), one might expect to observe singularities when the components of the tensors that describe the curvature of the spacetime diverge. This means that the curvature is higher than $\frac{1}{l_{p}^{2}}$, where $l_{p}$ is the Planck length, so we need to have into account the quantum effects, which are not considered in this theory. However, there are situations where this  behaviour is given as a result of the chosen coordinates. This is the case of the ``singularity" in $r=2M$ in the Schwarzschild metric. For this reason, another criteria, proposed by Penrose~\cite{Penrose}, is used to define a singularity: geodesic incompleteness. The physical interpretation of this condition is the existence of free falling observers that appear or disappear out of nothing. This is ``strange" enough to consider it a sufficient condition to assure that there is a singularity.

Already in the first solutions of Einstein equations there are ``places"  where the components of the curvature tensors diverge, like in $r=0$ in the Schwarzschild metric and $t=0$ in the Friedmann-Lema\^itre-Robertson-Walker (FLWR) metric, but it was thought that this was a consequence of the excessive symmetry of the solutions, as it occurs in many situations in classical mechanics or electromagnetism.
The first attempt of proving a singularity theorem was made by Raychaudhuri~\cite{Raychaudhuri} in 1955, in an article where he introduced his famous equation, which is essential in the later development of singularity theorems.
Ten years later, Penrose formulated the first singularity theorem that does not assume any symmetry~\cite{Penrose} (for a recent review see~\cite{SenPen}). It is also the first to use geodesic incompleteness in the definition of a singularity. This theorem showed that the singularity in $r=0$ of the Schwarzschild metric is also present under non symmetrical gravitational collapses.

The same happens with the singularity in $t=0$ of the FLRW metric, but this time is a consequence of a theorem stated by Hawking a year later~\cite{Hawking}, which predicts that, under three physically realistic conditions, all past directed timelike geodesics have finite length, therefore every particle of the Universe (hence the Universe itself) had a beginning. The mentioned conditions are that the action of the Ricci tensor over a timelike vector is greater or equal than cero, which it is interpreted as the attractive nature of gravity, that the Universe is globally hyperbolic and there is an hypersurface with positive initial expansion. Although we have said that the conditions are physically realistic, since it was measured the accelerated expansion of the Universe~\cite{Riess}, the convergence condition fails.

In general, all singularity theorems follow the same pattern, made explicit by Senovilla in~\cite{Senovilla}:
\begin{thm} 
(Pattern singularity ``theorem"). If the spacetime satisfies:\\
1) A condition on the curvature. \\
2) A causality condition. \\
3) An appropriate initial and/or boundary condition. \\
Then there are null or timelike inextensible incomplete geodesics. 
\end{thm}

Let us stop for a moment and analyse the configuration of the theorems. When the singularity theorems are derived, no assumptions are made on the underlying physical theory, that is, the one that links the matter and energy content with the structure of the spacetime. This means that they are valid, not only for GR, but for all the modifications that change the Einstein-Hilbert action. It is worth mentioning that the first condition can be reformulated using the field equations of the theory, obtaining what is known as the \emph{energy conditions}. These conditions are dependent of the considered theory, therefore they will differ from one to another, e.g., in GR they are formulated in terms of the energy-stress tensor only, while in f(R) theories there are some extra terms related to the curvature~\cite{ACdD}. 
Since we are working in a Lorentzian manifold, we have to endow it with an affine structure, which is implicitly assumed to be the Levi-Civita one, as it is postulated in GR, given by the Christoffel symbols~\cite{Wald},
\begin{equation}
\label{eq1}
\mathring{\Gamma}_{\mu\nu}^{\rho}=\frac{1}{2}g^{\rho\sigma}\left(\partial_{\mu}g_{\nu\sigma}+\partial_{\nu}g_{\mu\sigma}-\partial_{\sigma}g_{\mu\nu}\right).
\end{equation} 
This is the unique connection that is covariantly conserved~\cite{JimKo}, $\mathring{\nabla}_{\rho}g_{\mu\nu}=0$, and symmetric, $\mathring{\Gamma}_{\mu\nu}^{\rho}=\mathring{\Gamma}_{\nu\mu}^{\rho}$. \\
A metric has $D(D+1)/2$ components in a D-dimensional Lorentzian manifold, as it is a symmetric 2-covariant tensor. On the other hand, a general connection has $D^{3}$ components which are, in principle, completely independent degrees of freedom. Out of the $D^{3}$ components, $D^{2}(D-1)/2$ reside in the antisymmetric part
\begin{equation}
T_{\,\mu\nu}^{\rho}\equiv\Gamma_{\mu\nu}^{\rho}-\Gamma_{\nu\mu}^{\rho},
\end{equation} 
which is known as \emph{torsion}. The rest of degrees of freedom, $D^{2}(D+1)/2$, are encoded in the \emph{non-metricity} tensor
\begin{equation}
M_{\rho\mu\nu}=\nabla_{\rho}g_{\mu\nu}.
\end{equation} 
One might wonder if it is possible to modify the gravitational theory by setting these tensors to be different from zero, i.e. postulating a connection that it is not Levi-Civita. Certainly it is, although we have to take into account some considerations: 
\begin{enumerate}
\item Every connection assigns to a curve $\gamma$ a different acceleration, given by
\begin{equation}
a^{\nu}=v^{\mu}\nabla_{\mu}v^{\nu},
\end{equation} 
where $v^{\nu}=\frac{dx^{\mu}}{ds}$ is the four-velocity of the curve $\gamma$, parametrised by its proper time as $\gamma\left(s\right)=x^{\mu}\left(s\right)$.
If acceleration is to keep a meaning~\cite{TG}, it is necessary that the same metric is considered all along the curve. In other words, the connection must parallel-transport the metric, that is
\begin{equation}
v^{\rho}\nabla_{\rho}g_{\mu\nu}=0
\end{equation}
for every vector field $v^{\rho}$, which is equivalent to the metricity condition ($M_{\rho\mu\nu}=0$). This is why we will only consider connections that fulfill this condition from now on, although there has been work done in modified theories that set the non-metricity tensor different from zero, like in~\cite{JimKo} (for a review of these theories see~\cite{Nonmet}).
\item A different connection does not necessarily leads to a theory with a different phenomenology, since the action may be invariant or differ only by a divergence term under this change, therefore leaving the field equations unchanged. This is the case of a spacetime with linear vector distortion~\cite{JimKo} or teleparallel Gravity (TEGR)~\cite{TG}. 
\end{enumerate}
The latter case deserves some attention, as it is one of the simplest cases of this kind of theories, while at the same time, it is a good example to first apply the methods that we will use in more complicated ones. But first, let us review the singularity theorems in GR.

\section{Singularity theorems in General Relativity}
It seems logical that since we are generalizing the singularity theorems of GR, we introduce in this section the most general ones. This is the case of two recent theorems due to Senovilla and Galloway~\cite{SG}, that predict the occurrence of singularities, i.e. incomplete geodesics, based on the existence of trapped submanifolds of arbitrary co-dimension. The main key of the demonstration is, like in almost every singularity theorem, finding the conditions for the appearance of \emph{focal} and/or \emph{conjugate} points.

Let us consider a family of geodesics $\gamma_{s}\left(t\right)$, where $T^{\mu}=\left(\frac{\partial}{\partial t}\right)^{\mu}$ is the tangent vector to the family and $X^{\mu}=\left(\frac{\partial}{\partial s}\right)^{\mu}$ is the orthogonal deviation vector (that represents the displacement towards an infinitesimally near geodesic). These vectors follow the orthogonal deviation equation
\begin{equation}
T^{\mu}\nabla_{\mu}\left(T^{\nu}\nabla_{\nu}X^{\rho}\right)=-R_{\mu\nu\lambda}^{\,\,\,\,\,\,\rho}X^{\nu}T^{\mu}T^{\lambda}.
\end{equation}
A solution $X^{\mu}$ of this equation is called a \emph{Jacobi field on} $\gamma$. With this established we can see what we understand by conjugate and focal points:
\begin{defn}
Let $\gamma$ be a geodesic emanating from $p$ (orthogonal to a spacelike submanifold $\Sigma$). Then a point $q$ is conjugate (focal) along $\gamma$ to the the point $p$ (of the spacelike hypersurface $\Sigma$) if there exists a non-zero Jacobi field on $\gamma$ that vanishes at $p$ and $q$ (does not vanish at $\Sigma$ and vanishes at $q$).
\end{defn}
The problem of whether this kind of points will appear or not can be addressed in two different ways. In the physics orientated literature~\cite{HE, Wald} it is studied by means of the Raychaudhuri equation, which gives us the evolution of the expansion in a congruence of curves (not necessarily geodesics). To obtain this equation, we decompose the covariant derivative of the tangent vector of a congruence of curves, $B_{\mu\nu}=\mathring{\nabla}_{\nu}v_{\mu}$, into its antisymmetric $\omega_{\mu\nu}$, known as \emph{vorticity}, traceless symmetric $\sigma_{\mu\nu}$, usually referred as \emph{shear}, and trace part $\theta$, also known as \emph{expansion}, such as
\begin{equation}
B_{\mu\nu}=\frac{1}{3}\theta h_{\mu\nu}+\sigma_{\mu\nu}+\omega_{\mu\nu},
\end{equation}
where $ h_{\mu\nu}$ is the projection of the metric into the spacial subspace orthogonal to the tangent vector. Then, it can be seen that~\cite{HE}
\begin{eqnarray}
v^{\rho}\mathring{\nabla}_{\rho}\theta&=&\frac{d\theta}{ds}=-\frac{1}{3}\theta^{2}-\sigma^{\mu\rho}\sigma_{\mu\rho}
\nonumber
\\
&+&\omega^{\mu\rho}\omega_{\mu\rho}-\mathring{R}_{\rho\varphi}v^{\rho}v^{\varphi}+\mathring{\nabla}_{\mu}\left(v^{\nu}\mathring{\nabla}_{\nu}v^{\mu}\right),
\end{eqnarray}
which is the so called Raychaudhuri equation. With that, we can predict under what circumstances the expansion goes to minus infinity, which is the equivalent of having a conjugate/focal point~\cite{Wald}.   \\
On the other hand, in the mathematical literature~\cite{Oneill} this is solved in the context of variational calculus, by using the so-called \emph{Hessian} form. It is based on the idea that the set of all piecewise smooth curve segments $\gamma:\,\left[0,\, b\right]\longrightarrow M$ from a submanifold $P$ (that clearly includes the case $P=p$) to a point $q$, 
$\Omega\left(P,\, q\right)$, can be treated as a manifold. 

There is an explicit expression for this form, but before we write it we have to familiarize ourselves with the notation.
Let $\Sigma$ be a spacelike submanifold of arbitrary co-dimension, then we can define~\cite{SG}:
\begin{itemize}
\item $n_{\mu}$: future directed vector, perpendicular to the spacelike submanifold $\Sigma$.
\item $\overrightarrow{e}_{A}$: vector fields tangent to $\Sigma$.
\item $\gamma$: geodesic curve tangent to $n^{\mu}$ at $\Sigma$.
\item $u$: affine parameter along $\gamma$, taking $u=0$ at $\Sigma$.
\item $N^{\mu}$: geodesic vector field tangent to $\gamma$, having $\left.N_{\mu}\right|_{u=0}=n_{\mu}$.
\item $\overrightarrow{E}_{A}$: vector fields that are the parallel transport of $\overrightarrow{e}_{A}$ along $\gamma$ (using the Levi-Civita connection), satisfying that $\left.\overrightarrow{E}_{A}\right|_{u=0}=\overrightarrow{e}_{A}$.
\item $P^{\mu\nu}\equiv\gamma^{AB}E_{A}^{\mu}E_{B}^{\nu}$, where $\gamma^{AB}$ is the inverse of the first fundamental form of $\Sigma$ in the spacetime, $\gamma_{AB}=g_{\mu\nu}e_{A}^{\mu}e_{B}^{\nu}$. In $u=0$, $P^{\mu\nu}$ is just the proyector to $\Sigma$.
\item Expansion of the submanifold $\Sigma$ along $\overrightarrow{n}$: $\theta\left(\overrightarrow{n}\right)\equiv n_{\mu}H^{\mu}=\gamma^{AB}K_{AB}\left(\overrightarrow{n}\right)$, where $\overrightarrow{H}$ is the \emph{mean curvature vector} of the submanifold $\Sigma$, and $K_{AB}\left(\overrightarrow{n}\right)$ is the contraction of the \emph{shape tensor} $\overrightarrow{K}_{AB}$ with the one-form $n_{\mu}$~\cite{Oneill}. If $\theta\left(\overrightarrow{n}\right)<0$ for all posible normal vectors, $\Sigma$ is said to be a future trapped submanifold.
\end{itemize}
Now we can express the Hessian of two vector fields $V,\, W\in T_{\gamma}\left(\Omega\left(\Sigma,\, q\right)\right)$ as
\begin{equation}
\begin{array}{c}
I_{\gamma}\left(V,\, W\right)=\int_{0}^{b}\left[\left(N^{\mu}\nabla_{\mu}V^{\nu}\right)\left(N^{\rho}\nabla_{\rho}W_{\nu}\right)-\right.\\
\,\\
\left.-N_{\mu}R_{\,\,\nu\rho\sigma}^{\mu}V^{\nu}N^{\rho}W^{\sigma}\right]du+K_{AB}\left(\overrightarrow{n}\right)v^{A}w^{B},
\end{array}
\end{equation}
where $\overrightarrow{v}=\overrightarrow{V}\left(u=0\right)$, $v_{A}=v_{\mu}e_{A}^{\mu}$ is the part of $\overrightarrow{v}$  tangent to $\Sigma$ , and the same for $\overrightarrow{W}$ and $\overrightarrow{w}$~\cite{SG}.

The reader might be wondering what is the connection between the Hessian and the conjugate and focal points. The next theorem clears all doubts~\cite{Kriele}.
\begin{thm}
Let $\Sigma$ be a spacelike submanifold and $\gamma$ a causal curve orthogonal to $\Sigma$, then the submanifold $\Sigma$ does not have focal points along $\gamma$ if and only if the Hessian is semi-positive definite, having $I_{\gamma}\left(V,\, V\right)=0$ only if $\overrightarrow{V}$ is proportional to $\overrightarrow{N}$ on $\gamma$.
\end{thm}
To assure the appearance of focal points to a hypersurface of arbitrary co-dimension $\Sigma$, Senovilla and Galloway develop a curvature condition. 
\begin{prop}
Let $\Sigma$ be a spacelike submanifold of co-dimension $m$ in a Lorentzian manifold of dimension $n$, and let $n_{\mu}$ be a future-pointing normal to $\Sigma$. If $\theta\left(\overrightarrow{n}\right)\equiv\left(m-n\right)c<0$, and the curvature tensor satisfies the inequality
\begin{equation}
R_{\mu\nu\rho\sigma}N^{\mu}N^{\rho}P^{\nu\sigma}\geq0
\end{equation}
along $\gamma$, then there is a point focal to $\Sigma$ along $\gamma$ at or before $q=\gamma\left(u=\frac{1}{c}\right)$, given that the curve had arrived so far.
\end{prop}
This condition can be interpreted as a manifestation of the attractive character of gravity.

Based on this focalisation theorem, Senovilla and Galloway prove a generalisation of the Penrose and Hawking-Penrose theorem.
The first result predicts the incompleteness of null geodesics:
\begin{thm}
Let $\left(M, g\right)$ contain a non-compact Cauchy hypersurface $S$ and a closed future trapped submanifold $\Sigma$ of arbitrary co-dimension. If the curvature condition holds along every future directed null geodesic emanating orthogonally from $\Sigma$, then $(M, g)$ is future null geodesically incomplete.
\end{thm}
The second theorem is based  on the Hawking-Penrose lemma, which is valid for arbitrary dimension, that states that this three conditions cannot all hold: 
\begin{itemize}
\item Every inextensible causal geodesic contains a pair of conjugate points.
\item There are not closed timelike curves (chronology condition).
\item there is an achronal set $\Sigma$ such that $E^{+}\left(\Sigma\right)$ is compact.  
\end{itemize}
It is an established result~\cite{HE, Wald} that the first statement holds if $R_{\mu\nu}v^{\mu}v^{\nu}\geq0$ for every non-spacelike vector $v^{\mu}$. When applied to timelike vectors it is known as the \emph{timelike convergence condition}, while in the case of null ones it is called the \emph{null convergence condition}. Using the Einstein field equations we can rewrite these conditions in terms of the energy momentum tensor $T_{\mu\nu}$. The equivalent of the timelike convergence is the \emph{strong energy condition}, $T_{\mu\nu}v^{\mu}v^{\nu}\geq\frac{1}{2}T$, and for the null one the {weak energy condition}, $T_{\mu\nu}v^{\mu}v^{\nu}\geq0$, where $T$ is the trace of the energy-momentum tensor.

Now we can review the generalization of the H-P theorem:
\begin{thm}
If the chronology, generic, timelike and null convergence conditions hold and there is a closed future trapped submanifold $\Sigma$ of arbitrary co-dimension such that the curvature condition holds along every null geodesic emanating orthogonally from $\Sigma$, then the spacetime is causal geodesically incomplete.
\end{thm}
\emph{Sketch of the proof. } First of all, it has to be proven that the existence of closed trapped submanifolds leads to the existence of an achronal set with the properties mentioned in the lemma~\cite{SG}. Once the H-P lemma is proved, this theorem can be easily deduced, as it is explained in~\cite{HE}.

\section{Black hole regions}
We know from experience, e.g. the Schwarzschild metric, that the existence of incomplete null geodesics leads to the appearance of \emph{black holes}, that are regions of the spacetime that once an observer enters them, it cannot leave. This applies to all timelike and null curves, not just geodesics. This is usually known as the \emph{cosmic censorship conjecture}, which is a concept that Penrose introduced in 1969. It basically states that singularities cannot be \emph{naked}, that means that they cannot be seen by an outside observer. However, how can we express this concept mathematically?
The answer lies in the concept of \emph{conformal compactification}, which can be defined as~\cite{Infinity}:
\begin{defn}
Let $\left(M,g\right)$ and $\left(\tilde{M},\,\tilde{g}\right)$ be two spacetimes. Then $\left(\tilde{M},\,\tilde{g}\right)$ is said to be a conformal compactification of $M$ if and only if the following properties are met:
\begin{enumerate}
\item $M$ is an open submanifold of $\tilde{M}$ with smooth boundary $\partial\tilde{M}=\mathcal{J}$. This boundary is usually denoted \emph{conformal infinity}.
\item There exists a smooth scalar field $\Omega$ on $\tilde{M}$, such that $\tilde{g}_{\mu\nu}=\Omega^{2}g_{\mu\nu}$ on $M$, and so that $\Omega=0$ and its gradient $d\Omega\neq 0$ on $\mathcal{J}$.
\end{enumerate}
If additionally, every null geodesic in $M$ acquires a future and a past endpoint on $\mathcal{J}$, the spacetime is called \emph{asymptotically simple}. Also, if the Ricci tensor is zero in a neighbourhood of $\mathcal{J}$ the spacetime is said to be \emph{asymptotically empty}.
\end{defn} 
In a conformal compactification, $\mathcal{J}$  is composed by two null hypersurfaces, $\mathcal{J}^{+}$ and $\mathcal{J}^{-}$, known as \emph{future null infinity} and \emph{past null infinity} respectively.

In order to establish the definition of black hole, we need to introduce two more concepts~\cite{Wald}:
\begin{defn}
A spacetime $\left(M,g\right)$ is said to be \emph{asymptotically flat} if there is an asymptotically empty spacetime $\left(M',g'\right)$ and a neighbourhood $\mathcal{U}'$ of $\mathcal{J}'$, such that $\mathcal{U}'\cap M'$ is isometric to an open set $\mathcal{U}$ of $M$. 
\end{defn} 
\begin{defn}
Let  $\left(M,g\right)$ be an asymptotically flat spacetime with conformal compactification $\left(\tilde{M},\,\tilde{g}\right)$. Then $M$  is called \emph{(future) strongly asymptotically predictable} if there is an open region $\tilde{V}\subset\tilde{M}$, with $\overline{J^{-}\left(\mathcal{J}^{+}\right)\cap M}\subset\tilde{V}$, such that $\tilde{V}$ is globally hyperbolic. 
\end{defn}
This definition does not require the condition of the endpoints of the null geodesics, meaning that this kind of spacetimes can be singular. Nevertheless, if a spacetime is asymptotically predictable, then the singularities are not naked, i.e. are not visible from $\mathcal{J}^{+}$.

Now we can establish what we understand by a black hole:
\begin{defn}
A strongly asymptotically predictable spacetime $\left(M,g\right)$ is said to contain a black hole if $M$ is not contained in $J^{-}\left(\mathcal{J}^{+}\right)$. The \emph{black hole region}, $B$, is defined to be $B=M-J^{-}\left(\mathcal{J}^{+}\right)$ and its boundary, $\partial B$, is known as the \emph{event horizon}. 
\end{defn} 
Intuitively, we think that a particle in a closed trapped surface cannot scape to $\mathcal{J}^{+}$, meaning that it is part of the black hole region of the spacetime. Nevertheless, this is not true in general. In the next proposition we establish the conditions that ensure the existence of black holes when we have a closed future trapped submanifold of arbitrary co-dimension:
\begin{prop}
Let $\left(M,g\right)$ be a strongly asymptotically predictable spacetime of dimension $n$, and $\Sigma$ a closed future trapped submanifold of arbitrary co-dimension $m$ in $M$. If the curvature condition holds along every future directed null geodesic emanating orthogonally from $\Sigma$, then $\Sigma$ cannot intersect $J^{-}\left(\mathcal{J}^{+}\right)$, i.e. $\Sigma$ is in the black hole region $B$ of $M$ \footnote{Analogously, it can be defined a past strongly asymptotically predictable space time, and then the proposition would predict the existence of white hole regions,$B=M-J^{+}\left(\mathcal{J}^{-}\right)$, that are regions that particles cannot enter, only exit.}.
\end{prop} 
\begin{proof}
This proof is similar to the one of Proposition 12.2.2 by Wald~\cite{Wald}. Let us suppose that $\Sigma$ intersects $J^{-}\left(\mathcal{J}^{+}\right)$. Then, in the conformal compactification $\tilde{M}$, we would have that $J^{+}\left(\Sigma\right)\cap\mathcal{J}^{+}\neq\emptyset$. On other hand, we know that the spatial infinity $i^{0}$, the point of the compactification where the future (past) complete spacelike geodesics end (begin), is not in the causal future of any point in $M$. Therefore it follows trivially that $i^{0}\notin J^{+}\left(\Sigma\right)$. Since $M$ is strongly asymptotically predictable, there is a globally hyperbolic region $\tilde{V}$ in the compactification such that $\overline{J^{-}\left(\mathcal{J}^{+}\right)\cap M}\subset\tilde{V}$. From basic topology we have that the intersection of two closed sets is closed, therefore $\Lambda=\Sigma\cap\left(\overline{J^{-}\left(\mathcal{J}^{+}\right)\cap M}\right)$ is closed, where clearly $\Lambda\subset\Sigma$ and $\Lambda\subset\tilde{V}$. In addition, a closed subset of a compact is also compact, so from the compactness of $\Sigma$ we deduce that $\Lambda$ is compact. It is an standard result of Lorentzian geometry that, in a globally hyperbolic space, the causal future of a compact set is closed~\cite{Wald}, so, in this case, we have that $J^{+}\left(\Lambda\right)$ is closed in $\tilde{V}$. This means that it contains all of its limit points, therefore since $i^{0}\notin J^{+}\left(\Lambda\right)$, there is an open neighbourhood of $i^{0}$ that does not intersect $J^{+}\left(\Lambda\right)$, and so, an open region of $\mathcal{J}^{+}$ that does not intersect $J^{+}\left(\Lambda\right)$. It is known that a connected set cannot contain a subset with no boundary (except for the empty set and the set itself)~\cite{Munkres}. As we have already proved, $J^{+}\left(\Lambda\right)\cap \mathcal{J}^{+}$ is not equal to $\mathcal{J}^{+}$. Since $\mathcal{J}^{+}$ is connected, it follows that there must be a point $q\in \mathcal{J}^{+}$ in $\partial J^{+}\left(\Lambda\right)$. In the proof of the generalised Penrose theorem we used that in a globally hyperbolic spacetime $\partial J^{+}\left(\Lambda\right)=E^{+}\left(\Lambda\right)$, so in the compactification $\tilde{M}$ there is a null geodesic $\gamma$ connecting $p\in \Lambda\subset\Sigma$ with $q$. Furthermore, using the Theorem 51 of O'Neill~\cite{Oneill}, we see that this null geodesic must be orthogonal to $\Sigma$ and not contain any focal point of $\Sigma$ before $q$, as otherwise we would have that $q\in I^{+}\left(\Lambda\right)$ and therefore $q\notin E^{+}\left(\Lambda\right)$. With respect to the metric $g$ of $M$, $\gamma$ is also a null geodesic orthogonal to $\Sigma$ with no focal point of $\Sigma$, but now $\gamma$ is future complete~\cite{Wald}. Although, since $\Sigma$ is future trapped one has $\theta\left(\overrightarrow{n}\right)\equiv\left(m-n\right)c<0$ for any future-pointing null normal one-form $n_{\mu}$~\cite{SG}. Now, let $\left(m-n\right)C$ be the maximum value of all possible $\theta\left(\overrightarrow{n}\right)$ on the compact $\Sigma$. Then, using the Proposition II.3, we have that every null geodesic emanating orthogonally from $\Sigma$ will have a focal point at or before the affine parameter reaches the value $\frac{1}{C}$. This clearly leads to a contradiction, therefore the assumption is false.
\end{proof}
This Proposition will help us to study the singularities in theories of gravitation that include torsion. But first, let us introduce the main aspects of these theories.

\section{General aspects of theories with torsion}
In this section, we introduce the geometrical background of gravitational theories that allow a non symmetric connection that still fulfills the metricity condition. The interesting fact about these theories is that they appear naturally as a gauge theory of the Poincar\'e Group~\cite{Gauge}, making their formalism closer to that of the Standard Model of Particles, and hence making it a good candidate to explore the quantization of gravity.

Since the connection is not necessarily symmetric, the torsion can be different from zero.   
For an arbitrary connection, that meets the metricity condition, there exists a relation between it and the Levi-Civita connection
\begin{equation}
\label{eq2}
\mathring{\Gamma}_{\,\,\mu\nu}^{\rho}=\Gamma_{\,\,\mu\nu}^{\rho}-K_{\,\,\mu\nu}^{\rho},
\end{equation}
where 
\begin{equation}
K_{\,\,\mu\nu}^{\rho}=\frac{1}{2}\left(T_{\,\,\mu\nu}^{\rho}-T_{\mu\,\,\nu}^{\,\,\rho}-T_{\nu\,\,\mu}^{\,\,\rho}\right)
\end{equation}
is the \emph{contortion} tensor.

Since the curvature tensors depend on the connection, there is a relation between the ones defined throughout the Levi-Civita connection and the general ones. For the Riemann tensor we have~\cite{Shapiro}
\begin{equation}
\label{eq10}
\begin{array}{c}
\mathring{R}_{\,\,\mu\nu\rho}^{\sigma}=R_{\,\,\mu\nu\rho}^{\sigma}-\mathring{\nabla}_{\nu}K_{\,\,\mu\rho}^{\sigma}+\mathring{\nabla}_{\rho}K_{\,\,\mu\nu}^{\sigma}-\\
\,\\
-K_{\,\,\alpha\nu}^{\sigma}K_{\,\,\mu\rho}^{\alpha}+K_{\,\,\alpha\rho}^{\sigma}K_{\,\,\mu\nu}^{\alpha},
\end{array}
\end{equation}
where the upper index $\mathring{\,}$ denotes the Levi-Civita quantities.
By contraction we can obtain the expression for the Ricci tensor
\begin{equation}
\label{eq11}
\begin{array}{c}
\mathring{R}_{\mu\rho}=R_{\mu\rho}-\mathring{\nabla}_{\sigma}K_{\,\,\mu\rho}^{\sigma}+\mathring{\nabla}_{\rho}K_{\,\,\mu\sigma}^{\sigma}-\\
\,\\
-K_{\,\,\alpha\sigma}^{\sigma}K_{\,\,\mu\rho}^{\alpha}+K_{\,\,\alpha\rho}^{\sigma}K_{\,\,\mu\sigma}^{\alpha},
\end{array}
\end{equation}
and the scalar curvature
\begin{equation}
\begin{array}{c}
\mathring{R}=g^{\mu\rho}\mathring{R}_{\mu\rho}=R-\mathring{\nabla}^{\rho}K_{\,\,\sigma\rho}^{\sigma}-\\
\,\\
-K_{\,\,\alpha\sigma}^{\sigma}K_{\,\,\,\,\,\,\rho}^{\alpha\rho}+K_{\,\,\sigma\rho}^{\alpha}K_{\,\,\mu\alpha}^{\sigma}.
\end{array}
\end{equation}
All the theories that we will consider from now on will follow these geometrical properties, the only change would be the underlying physical theory.

\section{Singularities in Teleparallel Gravity}
TEGR is a degenerate case of the Poincar\'e gauge theories, since it is a gauge theory of the translation group only. Any gauge theory including these transformations will differ from the usual internal gauge models in many ways, the most significant being the presence of a tetrad field~\cite{TG2}. Given a nontrivial tetrad $h^{a}_{\,\mu}$, it is possible to define a connection known as Weitzenb\"{o}ck connection
\begin{equation}
\Gamma_{\,\,\mu\nu}^{\rho}=h_{a}^{\,\,\rho}\partial_{\mu}h_{\,\,\mu}^{a},
\end{equation}
that presents torsion, but no curvature. 
With this tetrad field we can also construct the Levi-Civita connection, taking into account that the metric can be expressed as
\begin{equation}
g_{\mu\nu}=\eta_{ab}h_{\,\,\mu}^{a}h_{\,\,\nu}^{b},
\end{equation}
where $\eta_{ab}$  is the Lorentz-Minkowski metric, and using the usual definition, as seen in Equation (\ref{eq1}).

The relation between these two connections is given by Equation (\ref{eq2}).
The Lagrangian density of this gravitational theory can be written as
\begin{equation}
\mathcal{L}=\frac{hc^{4}}{16\pi G}S^{\mu\nu\rho}T_{\mu\nu\rho},
\end{equation}
where $h=det\left(h_{\,\,\mu}^{a}\right)$, and
\begin{equation}
S^{\mu\nu\rho}=-S^{\mu\rho\nu}\equiv\frac{1}{2}\left(K^{\nu\rho\mu}+g^{\mu\rho}T_{\,\,\,\,\,\,\sigma}^{\sigma\nu}+g^{\mu\nu}T_{\,\,\,\,\,\,\sigma}^{\sigma\rho}\right),
\end{equation}
which is usually known as \emph{superpotential}.

Using the relation between the Weitzenb\"{o}ck and the Levi-Civita connection in Equation (\ref{eq2}) we can express this Lagrangian as
\begin{equation}
\mathcal{L}=\mathring{\mathcal{L}}-\partial_{\mu}\left(\frac{hc^{4}}{8\pi G}T_{\,\,\,\,\,\,\nu}^{\nu\mu}\right),
\end{equation}
where $\mathring{\mathcal{L}}$ is the Einstein-Hilbert Lagrangian of GR. Since they are equal except for a total divergence, the same field equations arise. Therefore it is a theory equivalent to GR, as it can be seen for example when one studies the junction conditions~\cite{Alvaro}.

The field equations can be obtained by taking variations of the Lagrangian. Expressing them in pure spacetime form, we have
\begin{equation}
\partial_{\sigma}\left(hS_{\mu}^{\,\,\sigma\nu}\right)-\frac{4\pi G}{c^{4}}\left(ht_{\mu}^{\,\,\nu}\right)=0,
\end{equation}
where
\begin{equation}
ht_{\mu}^{\,\,\nu}=\frac{hc^{4}}{4\pi G}\Gamma_{\,\,\rho\mu}^{\sigma}S_{\sigma}^{\,\,\rho\nu}+\delta_{\mu}^{\,\,\nu}\mathcal{L}
\end{equation}
is the canonical \emph{energy-momentum pseudotensor} of the gravitational field.
Although this is the simplest framework for a theory with torsion, it is helpful for introducing the methods that we will use in more general cases. In that sense, the next considerations are general, and can be applied in all the theories of gravitation.\\

In GR we have considered geodesic incompleteness as a criterium of the appearance of singularities, based on the fact that causal geodesics are the trajectories of free-falling observers. Therefore, we wish to modify this criteria by terms of these trayectories  in the theory that we are considering. We will say that our spacetime is singular if the domain of the affine parameter of at least one curve that follow any free-falling observer (including photons) is different from $\mathbb{R}$. For spacetimes in which we can define a conformal boundary, as the ones considered in Section III, this can be stated in the following way:
\begin{defn}
A spacetime $\left(M,g\right)$, endowed with a conformal compactification, is said to be singular if at least one non-spacelike curve has an endpoint outside the conformal infinity.
\end{defn}
Before continuing, it is useful to define two important classes of curves, which coincide in the case of the Levi-Civita connection~\cite{Hehl}:
\begin{itemize}
\item \textbf{Autoparallel curves}: these are the curves in which its tangent vector $v^{\mu}$ is parallel transported to itself, that is:
\begin{equation}
v^{\mu}\nabla_{\mu}v^{\nu}=0.
\end{equation}
The differential equation of the autoparallels is, under a suitable choice of the affine parameter:
\begin{equation}
\frac{dv^{\mu}}{dt}+\Gamma_{\rho\sigma}^{\mu}v^{\rho}v^{\sigma}=0,
\end{equation}
which only takes into account the symmetric part of the connection.
\item \textbf{Extremal curves}: these are the ones that extremise the length with respect to the metric of the manifold. It is worth mentioning that the length only depends on the metric, and not on the torsion. In order to see what are the equations of these curves we recall a standard result from Lorentzian geometry, that can be used as a definition:
\begin{thm}
Let $\gamma$ be a smooth timelike curve connecting two points $p,q\in M$.  Then the necessary and sufficient condition that $\gamma$ locally maximizes the length between $p$ and $q$ over smooth one parameter variations is that $\gamma$ is a geodesic with no point conjugate to $p$ between $p$ and $q$. 
\end{thm}
Then, the differential equations of these curves are the same of the Levi-Civita geodesics:
\begin{equation}
\frac{dv^{\mu}}{dt}+\mathring{\Gamma}_{\rho\sigma}^{\mu}v^{\rho}v^{\sigma}=0,
\end{equation}
 \end{itemize}
The trajectories of free-falling observers in theories different from GR do not follow these curves in general. Nevertheless, in TEGR they do. The equation of motion for free falling observers, scalar particles, is~\cite{TG2}:
\begin{equation}
\frac{dv_{\mu}}{dt}+\Gamma_{\rho\sigma\mu}v^{\rho}v^{\sigma}=0,
\end{equation}
which is equivalent to
\begin{equation}
\frac{dv^{\mu}}{dt}+\mathring{\Gamma}_{\rho\sigma}^{\mu}v^{\rho}v^{\sigma}=0.
\end{equation}
Therefore they follow extremal curves, which are the autoparallels of the Levi-Civita connection.

It is particular interesting to discuss this issue for photons. It has been stated that Maxwell equations do not couple to torsion in the minimal approach. However, in TEGR the electromagnetic field is able to couple to torsion without violating gauge invariance~\cite{TG}. Using the relation between the Levi-Civita and the Weitzenb\"{o}ck connection, one can verify that the teleparallel version of Maxwell's equations are completely equivalent with the usual Maxwell's equations in the context of GR. This means that they move according to the geodesic equation of GR, and so the causal structure is the same as in GR.

This discussion is more general. In fact, the equivalence between TEGR and GR means that all the singularity theorems developed in GR apply to this theory also. Therefore, the causal convergence and the curvature condition remain the same, although the expression for the Riemann and Ricci tensor change as discussed in the previous section, specifically in Equations (\ref{eq10}, \ref{eq11}).

\section{Singularities in Einstein-Cartan theory}
The Einstein-Cartan (EC) theory of gravitation is the most recognised theory that includes torsion~\cite{Sabbata,Hehl}. The main reason to introduce this theory is the fact that it allows to consider massive spinning fields in a natural way, while maintaining all the experimental success of GR. This theory arises when searching for a gravitational Lagrangian linear in the curvature term $R{^\lambda}\,_{\rho \mu \nu}$. The geometrical structure is the one analysed in section IV.

The field equations are obtained by varying the Lagrangian of this theory with respect to the metric and the contortion:
\begin{equation}
R_{\mu\nu}-\frac{1}{2}Rg_{\mu\nu}=\kappa\Sigma_{\mu\nu}
\end{equation}
and
\begin{equation}
\label{eq4}
S_{\mu\nu\rho}=\kappa\tau_{\mu\nu\rho},
\end{equation}
where
\begin{equation}
S_{\mu\nu}^{\,\,\,\,\,\rho}=T_{\mu\nu}^{\,\,\,\,\,\rho}+\delta_{\mu}^{\rho}T_{\nu\sigma}^{\,\,\,\,\,\sigma}-\delta_{\nu}^{\rho}T_{\mu\sigma}^{\,\,\,\,\,\sigma}
\end{equation}
is the modified torsion tensor.
At this point, we might wonder what are the trajectories of the free-falling observers, in order to establish some singularity theorems.

Since it is impossible to perform the minimally coupling prescription for the Maxwell's field while maintaining the $U\left(1\right)$ gauge invariance, the Maxwell equations are the same as in GR. Therefore, they move following null extremal curves, and so the causal structure is determined by the metric structure, just like in GR. Also, from the minimally couple procedure, it follows that particles with no spin, represented by scalar fields, do not feel torsion as well, since the covariant derivative of a scalar field is just its partial derivative. This means that the test particles follow the geodesics of the Levi-Civita connection, which allow us to generalise trivially the singularity theorems. Just like in TEGR, the causal convergence and the curvature conditions remain the same, it just changes the expression for the Levi-Civita Riemann and Ricci tensors, as given by Equations (\ref{eq10}, \ref{eq11}).

In any case, even for trajectories decoupled from torsion, energy conditions are modified.
Although the curvature condition is the same as in GR, these conditions change due to the fact that the field equations are different. Since Equation (\ref{eq4}) is purely algebraic we can substitute everywhere spin with torsion. Now we split the Einstein tensor into the Levi-Civita ($\mathring{G}^{\mu\nu}$) part and the rest, and we change the torsion terms by means of Equation (\ref{eq4}), obtaining
\begin{equation}
\label{eq5}
\mathring{G}^{\mu\nu}=\kappa\tilde{\sigma}^{\mu\nu},
\end{equation}
where $\tilde{\sigma}^{\mu\nu}$ is the combined energy-momentum tensor
\begin{equation}
\begin{array}{c}
\tilde{\sigma}_{\mu\rho}=\Sigma_{\mu\rho}-\nabla_{\sigma}K_{\,\,\mu\rho}^{\sigma}+\nabla_{\rho}K_{\,\,\mu\sigma}^{\sigma}-\\
\,\\
-K_{\,\,\alpha\sigma}^{\sigma}K_{\,\,\mu\rho}^{\alpha}+K_{\,\,\alpha\rho}^{\sigma}K_{\,\,\mu\sigma}^{\alpha}+\frac{1}{2}g_{\mu\rho}\left(\nabla^{\alpha}K_{\,\,\sigma\alpha}^{\sigma}+\right.\\
\,\\
\left.+K_{\,\,\alpha\sigma}^{\sigma}K_{\,\,\,\,\,\,\rho}^{\alpha\rho}-K_{\,\,\sigma\rho}^{\alpha}K_{\,\,\mu\alpha}^{\sigma}\right).
\end{array}
\end{equation}

Now, by using Equation (\ref{eq5}) we can write the energy conditions. The strong energy condition can be expressed as
\begin{equation}
\tilde{\sigma}_{\mu\nu}v^{\mu}v^{\nu}\geq\frac{1}{2}\tilde{\sigma},
\end{equation}
where $\tilde{\sigma}=g_{\mu\nu}\tilde{\sigma}^{\mu\nu}$. And for the weak energy condition we have
\begin{equation}
\tilde{\sigma}_{\mu\nu}v^{\mu}v^{\nu}\geq0.
\end{equation}
It is interesting noting that when the torsion is zero, one recovers the energy conditions of GR, as one would expect, since the contortion tensor involved in Equation (32) also vanishes.\\
So far we have analysed the singular behaviour of photons and spinless particles, but it is more interesting to study the behaviour of spinning fields. This question has already been addressed in the literature, mainly following two approaches. The first one is to study the singular behaviour of particular cosmological models using the energy conditions and the modified Raychaudhuri Equation for non symmetric connection derived by Stewart and Hajicek~\cite{Steward} (for a review of this approach see~\cite{Hehl}). These studies try to obtain plausible cosmological models that are singularity free. Nevertheless, they come to the conclusion that it is necesary to have regions with high spin density to observe a behaviour different from GR, and to avoid the singularities. On the other hand, Esposito~\cite{Esposito} proved a singularity theorem for EC theory based on the incompleteness of autoparallel curves. He considers this criteria to be sufficient to establish the singular character of a spacetime. \\
In those two approaches, the argument is based on the modified Raychaudhuri equation for non-symmetric metric connections. The main difference comes, as one would expect, from a change in the antisymmetric part of the decomposition mentioned in Equation (7), since now $B_{\mu\nu}$ is defined throughout the total connection. Then, the equation can be expressed as follows:
\begin{eqnarray}
&&v^{\rho}\nabla_{\rho}\theta=\frac{d\theta}{d\tau}=-\frac{1}{3}\theta^{2}-\sigma^{\mu\rho}\sigma_{\mu\rho}
\nonumber
\\&+&\left(\omega^{\mu\rho}+S^{\mu\nu}\right)\left(\omega_{\mu\rho}+S_{\mu\nu}\right)-R_{\rho\varphi}v^{\rho}v^{\varphi},
\end{eqnarray}
where $S_{\mu\nu}$ is a tensor that is usually defined through the following relation with the modified torsion tensor~\cite{Esposito}
\begin{equation}
S_{\mu\nu}^{\,\,\,\,\,\rho}=S_{\mu\nu}v^{\rho}.
\end{equation}
The problem with this reasoning is that the spin particles do not follow in general autoparallels curves of the total connection, so there might be situations where there is incompleteness of this kind of curves but no singular spin trajectories. Nevertheless, one could have the curiosity of knowing which is the Raychadhuri equation for the spin test particles. We will know study how we can expressing, making an study valid for all the Poincar\'e gauge theories of gravity.\\
All the analysis of these trajectories up to this point, which are reviewed in~\cite{ProbeB}, have a thing in common: after some algebra, they can all be expressed in the form  
\begin{eqnarray}
&&a^{\mu}=v^{\rho}\mathring{\nabla}_{\rho}v_{\mu}
\nonumber
\\
&=&C\left(\frac{\hbar}{m}\right)f\left(\widehat{R}_{\,\,\,\lambda\rho\sigma}^{\mu}s^{\rho\sigma}v^{\lambda}+K_{\rho\sigma}^{\,\,\,\,\,\,\mu}v^{\rho}v^{\sigma}\right),
\end{eqnarray}
where $C$ is a constant, $m$ is the mass of the particle, and we have made explicit the Planck constant. The tensor $s^{\rho\sigma}$ is the internal spin tensor, related to the spin $s^{\mu}$ of the particle by
\begin{equation}
s^{\mu}=\frac{1}{2}\epsilon^{\mu\nu\rho\sigma}v_{\nu}s_{\rho\sigma},
\end{equation}
where $\epsilon^{\mu\nu\rho\sigma}$ is the totally antisymmetric Levi-Civita tensor, which is normalised with the square root of the metric, as it is usual in a Lorentzian manifold. The \emph{function} $f$ represents a linear combination of different contractions of the tensors involved in the expression, depending on the analysis chosen. The connection with respect it is calculated the Riemmann tensor in brackets is also dependent of the analysis, but it is always one constructed with the Levi Civita and linear combinations of torsion related quantities.  \\
When writing the Raychaudhuri equation we choose to make it with respect to the Levi-Civita connection, since it is analogous to the expression in terms of the total connection, and in this way we avoid introducing new terms to the decomposition in Equation (7). With that in mind we have that  
\begin{eqnarray}
&&v^{\rho}\mathring{\nabla}_{\rho}\theta=\frac{d\theta}{ds}=-\frac{1}{3}\theta^{2}-\sigma^{\mu\rho}\sigma_{\mu\rho}
\nonumber
\\
&&+\omega^{\mu\rho}\omega_{\mu\rho}-\mathring{R}_{\rho\varphi}v^{\rho}v^{\varphi}
\nonumber
\\
&&+C\left(\frac{\hbar}{m}\right)\mathring{\nabla}_{\mu}\left(\widehat{R}_{\,\,\,\lambda\rho\sigma}^{\mu}s^{\rho\sigma}v^{\lambda}+K_{\rho\sigma}^{\,\,\,\,\,\,\mu}v^{\rho}v^{\sigma}\right).
\end{eqnarray}
Using this equation we could predict the appearance of focal/conjugate points in a congruence of this timelike curves, just by imposing a generalised curvature condition
\begin{equation}
\mathring{R}_{\rho\varphi}v^{\rho}v^{\varphi}-C\left(\frac{\hbar}{m}\right)\mathring{\nabla}_{\mu}\left[f\left(\widehat{R}_{\,\,\,\lambda\rho\sigma}^{\mu}s^{\rho\sigma}v^{\lambda}+K_{\rho\sigma}^{\,\,\,\,\,\,\mu}v^{\rho}v^{\sigma}\right)\right]\geq0,
\end{equation}
that must hold for every timelike vector $v^{\mu}$. Nevertheless, there are some issues with this approach. First of all, this is only valid for congruences of curves that have the same spin orientation for all the test particles, hence limiting the analysis. On the other hand, we know for the singularity theorems in GR that the existence of focal/conjugate points is not a sufficient condition for the appearance of singularities. We also need global conditions that allow us to reach a contradiction with the completeness of the curves. Since we are considering non-geodesical behaviour, the theorems that allow us to make that contradiction are no longer valid, and we cannot predict the singularities. That is why in this article we propose another approach, based on the result of the appearance of black/white hole regions in an arbitrary Lorentzian manifold. \\

It is clear from the previous analysis that spinning particles do not follow extremal curves. However, independently of how torsion affects these particles, they will follow timelike curves, since they are massive and we assume that locally (in a normal neighbourhood of a point) nothing can be faster than light (null geodesics). Hence, it would be interesting to see under what circumstances we have non-geodesical timelike singularities. For that, we recover the definition of an n-dimensional black/white hole given in Section III. From this definition, we conclude that if this kind of structures exist in our spacetime, we would have timelike curves (including non-geodesics) that do not have endpoints in the conformal infinity, since for the case of black holes, the spacetime $M$ is not contained in $J^{-}\left(\mathcal{J}^{+}\right)$, while for white holes, $M$ is not contained in $J^{+}\left(\mathcal{J}^{-}\right)$. This is exactly the extended definition of singularity that we have given in Section V. Considering this, we establish the following theorem:
\begin{thm}
Let $\left(M,g\right)$ be a strongly asymptotically predictable spacetime of dimension $n$, and $\Sigma$ a closed future trapped submanifold of arbitrary co-dimension $m$ in $M$. If the curvature condition holds along every future directed null geodesic emanating orthogonally from $\Sigma$, then some timelike curves in $M$ would not have endpoints in the conformal infinity, hence $M$ is a singular spacetime (Definition V.1).
\end{thm} 

From a physical point of view, one might wonder if one of the incomplete timelike curves actually represents the trajectory of a spin particle. From Equation (37), which represents the non-geodesical behaviour, we see that the only possible way that all the trajectories have endpoints in the conformal infinity is that there are huge values of the curvature and torsion near the event horizon, which in a physically plausible scenario it is not possible. This is why we consider it a more physically relevant theorem for the singular behaviour of the spin particles, since it is strongly related to the actual trajectories.

\section{Singularities in dynamical torsion theories}

So far, the two torsion theories that we have analysed are part of a set of theories known as Poincar\'e Gauge Gravity (PG)~\cite{Hehl, Poincare}. The reason why there are many theories under this premise is because we can construct a large number of invariants from the curvature and torsion tensors, and therefore a general gravitational Lagrangian has the complicated form of a sum of all available invariants of proper dimension. The coefficients in the sum can be arranged to obtain different gravitational theories (for some criteria on the election and stability of a large class of these PG Lagrangians see~\cite{Poincare,Sezgin}).

In this section, we will study a PG theory of gravity, hence it has the same geometrical background explained in section IV, with the following vacuum Lagrangian~\cite{JJ}:
\begin{equation}
\begin{array}{c}
S=\frac{c^{4}}{16\pi G}\int d^{4}x\sqrt{-g}\left[2\Lambda-\mathring{R}-\right.\\
\,\\
-\frac{1}{2}\left(2c_{1}+c_{2}\right)R_{\lambda\rho\mu\nu}R^{\mu\nu\lambda\rho}+c_{1}R_{\lambda\rho\mu\nu}R^{\lambda\rho\mu\nu}+\\
\,\\
\left.+c_{2}R_{\lambda\rho\mu\nu}R^{\lambda\mu\rho\nu}+d_{1}R_{\mu\nu}\left(R^{\mu\nu}-R^{\nu\mu}\right)\right].
\end{array}
\end{equation}
One interesting feature about this theory is that if we set the torsion to be zero, we recover GR. Then, it is expected to involve slight modifications to the standard theory in terms of the torsion tensor alone.

The field equations can be derived from this action by performing variations with respect to the gauge potentials $A_{\mu}$, which are related to the translations and Lorentz rotations generators of the Poincar\'e group in the following way:
\begin{equation}
A_{\mu}=e_{\,\,\,\mu}^{a}P_{a}+\omega_{\,\,\,\,\,\,\mu}^{ab}J_{ab},
\end{equation}
where $e_{\,\,\,\mu}^{a}$ is the vierbein field and $\omega_{\,\,\,\,\,\,\mu}^{ab}$ is the spin connection~\cite{JJ}.
The generators $P_{a}$ and $J_{ab}$ follow the usual commutation relations:
\begin{equation}
\left[P_{a},\,P_{b}\right]=0,
\end{equation}
\begin{equation}
\left[P_{a},\,J_{bc}\right]=i\eta_{a\left[b\right.}P_{\left.c\right]},
\end{equation}
\begin{equation}
\left[J_{ab},\,J_{cd}\right]=\frac{i}{2}\left(\eta_{ad}J_{bc}+\eta_{cb}J_{ad}-\eta_{db}J_{ac}-\eta_{ac}J_{bd}\right).
\end{equation} 
With that procedure we obtain the field equations:
\begin{equation}
\label{eq6}
\begin{array}{c}
\mathring{G}_{\mu}^{\,\,\nu}=-\frac{1}{2}\Lambda\delta_{\mu}^{\,\,\nu}+2c_{1}T1_{\mu}^{\,\,\nu}+c_{2}T2_{\mu}^{\,\,\nu}-\\
\,\\
-\left(2c_{1}+c_{2}\right)T3_{\mu}^{\,\,\nu}+d_{1}\left(H1_{\mu}^{\,\,\nu}-H2_{\mu}^{\,\,\nu}\right)
\end{array}
\end{equation}
and
\begin{equation}
\begin{array}{c}
2c_{1}C1_{\left[\mu\lambda\right]}^{\,\,\,\,\,\,\,\nu}-c_{2}C2_{\left[\mu\lambda\right]}^{\,\,\,\,\,\,\,\nu}+\left(2c_{1}+c_{2}\right)C3_{\left[\mu\lambda\right]}^{\,\,\,\,\,\,\,\nu}-\\
\,\\
-d_{1}\left(Y1_{\left[\mu\lambda\right]}^{\,\,\,\,\,\,\,\nu}-Y2_{\left[\mu\lambda\right]}^{\,\,\,\,\,\,\,\nu}\right)=0,
\end{array}
\end{equation}
where the functions $T,H,C,Y$ depend on the Riemann and torsion tensor and their contractions:
\begin{equation}
\begin{array}{c}
\,\\
T1_{\mu}^{\,\,\,\nu}\equiv R_{\lambda\rho\mu\sigma}R^{\lambda\rho\nu\sigma}-\frac{1}{4}\delta_{\mu}^{\,\,\,\nu}R_{\lambda\rho\alpha\sigma}R^{\lambda\rho\alpha\sigma},\\
\,\\
T2_{\mu}^{\,\,\,\nu}\equiv R_{\lambda\rho\mu\sigma}R^{\lambda\nu\rho\sigma}+R_{\lambda\rho\sigma\mu}R^{\lambda\sigma\rho\nu}-\\
-\frac{1}{2}\delta_{\mu}^{\,\,\,\nu}R_{\lambda\rho\alpha\sigma}R^{\lambda\alpha\rho\sigma},\\
\,\\
T3_{\mu}^{\,\,\,\nu}\equiv R_{\lambda\rho\mu\sigma}R^{\nu\sigma\lambda\rho}-\frac{1}{4}\delta_{\mu}^{\,\,\,\nu}R_{\lambda\rho\alpha\sigma}R^{\alpha\sigma\lambda\rho},\\
\,\\
H1_{\mu}^{\,\,\,\nu}\equiv R_{\,\,\,\lambda\mu\rho}^{\nu}R^{\lambda\rho}+R_{\lambda\mu}R^{\lambda\nu}-\frac{1}{2}\delta_{\mu}^{\,\,\,\nu}R_{\lambda\rho}R^{\lambda\rho},\\
\,\\
H2_{\mu}^{\,\,\,\nu}\equiv R_{\,\,\,\lambda\mu\rho}^{\nu}R^{\rho\lambda}+R_{\lambda\mu}R^{\nu\lambda}-\frac{1}{2}\delta_{\mu}^{\,\,\,\nu}R_{\lambda\rho}R^{\rho\lambda},\\
\,\\
C1_{\mu}^{\,\,\,\,\lambda\nu}\equiv\mathring{\nabla}_{\rho}R_{\mu}^{\,\,\,\,\lambda\rho\nu}+K_{\,\,\,\,\sigma\rho}^{\lambda}R_{\mu}^{\,\,\,\,\sigma\rho\nu}-K_{\,\,\,\,\mu\rho}^{\sigma}R_{\sigma}^{\,\,\,\,\lambda\rho\nu},\\
\,\\
C2_{\mu}^{\,\,\,\,\lambda\nu}\equiv\mathring{\nabla}_{\rho}\left(R_{\mu}^{\,\,\,\,\nu\lambda\rho}-R_{\mu}^{\,\,\,\,\rho\lambda\nu}\right)+\\
+K_{\,\,\,\,\sigma\rho}^{\lambda}\left(R_{\mu}^{\,\,\,\,\nu\sigma\rho}-R_{\mu}^{\,\,\,\,\rho\sigma\nu}\right)-\\
-K_{\,\,\,\,\mu\rho}^{\sigma}\left(R_{\sigma}^{\,\,\,\,\nu\lambda\rho}-R_{\sigma}^{\,\,\,\,\rho\lambda\nu}\right),\\
\,\\
C3_{\mu}^{\,\,\,\,\lambda\nu}\equiv\mathring{\nabla}_{\rho}R_{\,\,\,\,\,\,\,\,\,\mu}^{\rho\nu\lambda}+K_{\,\,\,\,\sigma\rho}^{\lambda}R_{\,\,\,\,\,\,\,\,\,\mu}^{\rho\nu\sigma}-K_{\,\,\,\,\mu\rho}^{\sigma}R_{\,\,\,\,\,\,\,\,\,\sigma}^{\rho\nu\lambda},\\
\,\\
Y1_{\mu}^{\,\,\,\,\lambda\nu}\equiv\delta_{\mu}^{\,\,\,\nu}\mathring{\nabla}_{\rho}R^{\lambda\rho}-\mathring{\nabla}_{\mu}R^{\lambda\nu}+\delta_{\mu}^{\,\,\,\nu}K_{\,\,\,\sigma\rho}^{\lambda}R^{\sigma\rho}+\\
+K_{\,\,\,\mu\rho}^{\rho}R^{\lambda\nu}-K_{\,\,\,\mu\rho}^{\nu}R^{\lambda\rho}-K_{\,\,\,\rho\mu}^{\lambda}R^{\rho\nu},\\
\,\\
Y2_{\mu}^{\,\,\,\,\lambda\nu}\equiv\delta_{\mu}^{\,\,\,\nu}\mathring{\nabla}_{\rho}R^{\rho\lambda}-\mathring{\nabla}_{\mu}R^{\nu\lambda}+\delta_{\mu}^{\,\,\,\nu}K_{\,\,\,\sigma\rho}^{\lambda}R^{\rho\sigma}+\\
+K_{\,\,\,\mu\rho}^{\rho}R^{\nu\lambda}-K_{\,\,\,\mu\rho}^{\nu}R^{\rho\lambda}-K_{\,\,\,\rho\mu}^{\lambda}R^{\nu\rho}.\\
\,
\end{array}
\end{equation}
As we have explained, the only difference between this theory and EC are the fields equations. This means that the curvature conditions remain the same, and so does the Proposition about the appearance of black holes. Nevertheless, the energy conditions change.

In Equation (\ref{eq6}) we have already isolated the Levi-Civita Einstein tensor $\mathring{G}$, therefore we can consider the right side of the equation as an effective energy-momentum tensor
\begin{equation}
\mathring{G}_{\mu\nu}=\mathcal{T}_{\mu\nu}.
\end{equation} 
This leads us to the energy conditions for this theory:
\begin{itemize}
\item Strong energy condition:
\begin{equation}
\mathcal{T}_{\mu\nu}v^{\mu}v^{\nu}\geq\frac{1}{2}\mathcal{T}
\end{equation}
for every timelike vector $v^{\mu}$. 
\item Weak energy condition:
\begin{equation}
\mathcal{T}_{\mu\nu}v^{\mu}v^{\nu}\geq0
\end{equation}
for every null vector $v^{\mu}$. 
\end{itemize}

These conditions depend on some intricate functions of the curvature tensor, and it makes us think that probably it is better in this case (and also in EC) to evaluate the conditions directly calculating the torsion-free Riemann and Ricci tensor of the considered metric. However, expressing them in this form makes us realise of some curious facts about the theory.

It is interesting to note that in GR a vacuum solution always meets the energy conditions. In this theory though, the situation is different. For example, we can arrange the coefficients in a way that the spacetime contains a closed trapped submanifold of codimension 2 (closed trapped surface) and yet be a singularity free spacetime. This is impossible for a vacuum solution in GR (if the generic condition holds), since in this kind of solutions the Ricci tensor is identically zero. 

Let us now explore a specific case. First, we set all the coefficients to zero except for $d_{1}$. Observing the field equations, we see that the second one can be solved by setting the Ricci tensor to be zero. In that case, the first equation is just:
\begin{equation}
\mathring{G}_{\mu\nu}=0,
\end{equation}
which is the vacumm field equation in GR. This means that flat Ricci solutions ($R_{\mu\nu}=0$) recover the same metrics that GR. However, this is not true for an arbitrary connection, since the equations that relate the Ricci tensor with the Levi-Civita one must hold. Therefore, this statement would be true for connections that follow the equation
\begin{equation}
\mathring{\nabla}_{\sigma}K_{\,\,\mu\rho}^{\sigma}-\mathring{\nabla}_{\rho}K_{\,\,\mu\sigma}^{\sigma}+K_{\,\,\alpha\sigma}^{\sigma}K_{\,\,\mu\rho}^{\alpha}-K_{\,\,\alpha\rho}^{\sigma}K_{\,\,\mu\sigma}^{\alpha}=0.
\end{equation} 
At first sight, one might think that the only solution to this equation is a zero contortion tensor, hence obtaining a torsion-free spacetime. However, let us for example take $K_{\,\,10}^{0}=-K_{\,\,00}^{1}=1$ and the rest to be zero. Then it is easy to see that the previous equation holds. Therefore, with a suitable connection we can recover all the metrics of the vacuum solutions of GR in a torsion theory.

The interesting fact is that, although the metrics are the same as in GR, and hence very well known spacetimes that describe satisfactorily many physical situations, the underlying theory is different, and so the matter and energy content and the motion of particles will differ from GR. Nevertheless, as we have seen, we can still apply the GR singularity theorems to scalar fields and photons, and the black hole formalism for the rest of particles. Since the metric is the same, the conditions of the appearance of timelike and null singularities and black/white hole regions would be the same as in GR. So in this case, we can establish that the presence of torsion does not change the singular behaviour of the spacetime.

Although this was a rather special case, it is possible to recover some famous metrics with a more general election of the coefficients. This is the case of a recent solution by two of the authors~\cite{JJ}, where a Reissner-Norstr\"{o}m solution is found setting the coefficients to be $c_{1}=-d_{1}/4$ and $c_{2}=-d_{1}/2$. Since this is a black hole solution, we can study the singular behaviour of spin particles within this framework. 

\section{Conclusions}
In this work we have studied how to extend the tools used in GR to deduce the appearance of singularities to theories of gravitation that include torsion. In order to study that, we have first reviewed two modern singularity theorems by Senovilla and Galloway. For our purposes, the interesting part about these theorems is the curvature condition that they obtain to predict the existence of focal points of a spacelike submanifold. We have used that result to prove the Proposition III.5, that gives us the necessary conditions for the appearance of black/white hole regions of arbitrary dimension in a spacetime. With that established, we have analysed three particular theories. In the case of TEGR we have obtained equivalent results to GR, although the expression for the curvature tensors change, as one might expect. In EC theory we have seen that for minimally coupled scalar fields and photons we can use the results proved in GR. For the rest of particles, we consider the existence of black/white hole as an indicator of the singular character of their trajectories. In this case we also obtain their energy conditions. For the dynamical torsion example we have made a similar analysis of that of EC theory. We have obtained the same geometrical results, although the energy conditions change, leading to some interesting behaviours. For instance, we have shown that in a vacuum solution we can have a violation of the energy conditions, something that cannot happen in GR or EC theory. Furthermore, we analyse a particular Lagrangian and obtain that we can reproduce all the metric structure of the vacuum solutions of GR in theories with torsion. \\
The formalism that we have developed can be used in other modified gravity theories, as long as the inner structure is a Lorentzian manifold, using the following considerations. As we have already discussed, a minimally coupled scalar field in these theories will follow timelike geodesics, so we can use the singularity theorems of GR that are based on incomplete timelike geodesics, such as the Hawking theorem. On the other hand, we have been using the fact that in the theories that we have considered, photons follow null geodesics. This is not necessarily true for all the torsion theories, since in some of them we can couple the Maxwell equations to torsion non-minimally and still preserve the gauge invariance~\cite{Hojman}. Nevertheless, this would mean that we can still use the black hole formalism, because they would not follow spacelike curves. Here we can see how powerful this result is, because it allows us to predict the singular behaviour of any non-spacelike curve, which includes coupled photons, spinning particles or non-minimal coupled fields.\\
Moreover, we have used the cosmic censorship as a plausible condition in torsion theories. In any case, it would be very interesting to study the possible creation of naked singularities in these theories under physical realistic conditions~\cite{delacruz}, and to test with concrete examples if spinning particles would reach the black/white hole regions. In order to conclude if the spin can advert singularities in torsion theories, it is useful to work in the semiclassical limit of the Dirac wave function via the WKB approximation, as treated in~\cite{WKB}. Using the equation of motion given by Audretsch we can simulate numerically the movement of spin particles around the event horizon. Work is in progress along this line.\\

\acknowledgments
This work has been supported in part by the MINECO (Spain) projects FIS2014-52837-P, FPA2014-53375-C2-1-P, and Consolider-Ingenio MULTIDARK CSD2009-00064.

\end{document}